\documentclass[twoside]{article}

\usepackage{qic}
\usepackage{amsmath,amssymb,amsfonts,mathtools}
\usepackage{amsthm}
\usepackage{mathrsfs}
\usepackage{enumitem}
\usepackage[hidelinks]{hyperref}
\usepackage{tikz}
\usetikzlibrary{arrows.meta,calc,decorations.pathreplacing}

\theoremstyle{plain}
\newtheorem{theorem}{Theorem}[section]
\newtheorem{lemma}[theorem]{Lemma}
\newtheorem{proposition}[theorem]{Proposition}
\newtheorem{corollary}[theorem]{Corollary}

\theoremstyle{definition}
\newtheorem{definition}[theorem]{Definition}

\newtheorem{conjecture}[theorem]{Conjecture}

\theoremstyle{remark}
\newtheorem{remark}[theorem]{Remark}

\newcommand{\F}{F_2}
\newcommand{\Z}{\mathbb{Z}}

\newcommand{\R}{\mathbb{R}}

\newcommand{\tr}{\mathrm{tr}}
\newcommand{\ab}{\mathrm{ab}}
\newcommand{\laurent}{\Z[T^{\pm1},S^{\pm1}]}

\newcommand{\angles}[1]{\left\langle #1\right\rangle}
\newcommand{\Jset}{\mathcal J}
\newcommand{\runningtitle}{Exact Separation via Trace Geometry}
\newcommand{\runningauthors}{Z. Chen and J. Wu}

\begin{document}
\setlength{\textheight}{8.0truein}
\runninghead{\runningtitle}{\runningauthors}

\normalsize\textlineskip
\thispagestyle{empty}
\setcounter{page}{1}
\copyrightheading{0}{0}{2026}{000--000}

\vspace*{0.88truein}

\fpage{1}

\centerline{\bf EXACT SEPARATION OF WORDS}
\vspace*{0.035truein}
\centerline{\bf VIA TRACE GEOMETRY}
\vspace*{0.37truein}
\centerline{\footnotesize ZEYU CHEN\footnote{Email: chenzeyu@zju.edu.cn} AND JUNDE WU\footnote{Email: wjd@zju.edu.cn}}
\vspace*{0.015truein}
\centerline{\footnotesize\it School of Mathematical Sciences, Zhejiang University}
\baselineskip=10pt
\centerline{\footnotesize\it Hangzhou 310058, People's Republic of China}
\vspace*{0.225truein}
\publisher{ }{ }

\vspace*{0.21truein}

\abstracts{A basic question in the study of measure-once quantum finite automata is whether two distinct input words can be separated with certainty. The exact separation problem reduces to a trace-vanishing question in \(SU(2)\). The main difficulty lies in the genuinely nonabelian regime, where \(u\) and \(v\) have the same abelianization. This paper develops a slice-driven framework that converts algebraic invariants of the word---prefix statistics, metabelian polynomials, and slope specializations---into explicit low-dimensional families in \(SU(2)^2\) on which the trace-vanishing question can be analyzed effectively. A quadratic trace-deficit identity on a principal one-parameter family provides the main algebraic-to-geometric bridge. Building on this framework, the paper establishes three core certified slice criteria: a dihedral criterion, equivalently readable through a signed \(a\)-count; a quaternionic criterion; and a local one-row criterion. Together with a supplementary interior-point test and a binary-dihedral slice, these results sharply reduce the unresolved portion of the problem to a residual super-degenerate class, while also clarifying the limitations of certification strategies based only on finitely many finite-subgroup evaluations.}{}{}

\vspace*{10pt}

\vspace*{3pt}

\vspace*{1pt}\textlineskip
\section{Introduction}
Quantum computing has advanced rapidly in recent years, with experimental platforms, algorithmic ideas, and complexity-theoretic viewpoints developing side by side~\cite{Shor1997,Preskill2018,Arute2019,Kim2023}. Within that broad development, quantum finite automata (QFAs) remain one of the most revealing models. Their state space is tiny, their dynamics are transparent, and yet the basic quantum ingredients of interference, reversibility, and measurement already appear in full force. Since the foundational works of Moore and Crutchfield on measure-once quantum finite automata (MO-QFAs) and of Kondacs and Watrous on measure-many and two-way variants, QFAs have served both as a compact testing ground for quantum advantage and as a precise way to understand what finite quantum memory can and cannot do~\cite{MooreCrutchfield00,KondacsWatrous97,AmbainisFreivalds98,BrodskyPippenger02,SayYak14}. For exact and zero-error questions in particular, the model is especially attractive: it is simple enough to allow explicit matrix calculations, but rigid enough that every successful construction reflects a genuinely structural phenomenon.

Among these questions, the \emph{exact separation problem} stands out for its conceptual simplicity: given two distinct words \(u\) and \(v\) over a binary alphabet, can a quantum automaton accept one and reject the other with certainty? In the broader separation program of Belovs, Montoya, and Yakary{\i}lmaz~\cite{Choffrut17,BMY16}, two-state QFAs were shown to separate every word pair in the nondeterministic mode, while the universal two-state zero-error question was left as an explicit conjecture. The present paper addresses the geometric core of that remaining question.

In the two-state measure-once QFA (MO-QFA) model, exact separation reduces to a remarkably concrete problem~\cite{BMY16}: it suffices to produce matrices \(A,B\in SU(2)\) for which a certain trace vanishes. A question from quantum automata theory thus becomes a question about the \emph{trace geometry} of word maps on a compact Lie group. When the two words have different letter counts, i.e.\ \(\ab(u)\neq\ab(v)\), the answer is already visible at the abelian level. The genuine difficulty begins in the \emph{hard} regime, where both words are positive (containing only \(a\) and \(b\), no inverses) and share the same abelianization, so that
\[
 w:=u^{-1}v\in F_2'.
\]
This positive-word-difference formulation isolates the first genuinely nonabelian layer of the problem while preserving enough combinatorial rigidity to support explicit analysis~\cite{BMY16}. The work of~\cite{Choffrut17} is decisive in identifying this conjectural frontier, but it does not resolve the equal-abelianization case for two-state MO-QFAs, which is precisely the regime studied here.

Even within their narrow state space, exact QFAs can be strikingly economical: for suitable promise problems, they achieve constant size where classical automata require unbounded growth~\cite{AmbainisYakaryilmaz12,GruskaQiuZheng15}. At the same time, MO-QFAs are tightly constrained, and even their bounded-error power forms only a proper subclass of the regular languages~\cite{BrodskyPippenger02}. The exact separation question therefore sits at a delicate boundary, where many word families yield to algebraic invariants or special matrix constructions, but the most symmetric hard words continue to resist the available criteria.

On the group-theoretic side, recent work on global word maps provides important context. Khoi and Toan~\cite{KhoiToan22} gave a trace-polynomial criterion for surjectivity of certain word maps on \(SU(2)\) and used it to classify both surjective and non-surjective families. However, exact separation requires only that the image of the word map meet the trace-zero locus, not that it cover all of \(SU(2)\), so a surjectivity-based strategy is overly restrictive for exact separation. A more targeted geometric organizing principle is therefore needed.

This paper develops a \emph{slice-driven} approach, guided by a structural obstruction: no strategy that evaluates words only on a fixed finite catalog of finite subgroups can be universal for hard positive-word differences. A slice is a concrete low-dimensional family of pairs \((A,B)\in SU(2)^2\) on which the trace of \(w(A,B)\) can be read from explicit invariants of the word. A slice family is a collection of such slices, typically indexed by a discrete parameter such as a slope or a local exponent, and the parameter choice is allowed to depend on the word.

Rather than searching the full representation space \(SU(2)^2\), the framework restricts attention to explicit manifolds where the trace geometry is governed by metabelian data. On the algebraic side, the metabelian polynomial \(M_w(T,S)\) is decomposed into interval blocks indexed by prefix statistics of positive words, and suitable slope specializations \(S=T^r\) preserve nontrivial information. On the geometric side, a slope-visible slice converts that specialization into a quadratic trace-deficit identity near commuting basepoints, providing a direct algebraic-to-geometric bridge from Fox calculus to guaranteed trace drops toward the trace-zero locus.

Building on this functional engine, the paper assembles a certified exact-separation menu whose core consists of three slice criteria (Theorem~\ref{thm:certified-menu-items}):
\begin{itemize}[leftmargin=*,itemsep=2pt]
\item \textbf{The dihedral/signed-\(a\)-count criterion:} a global evaluation on a mixed normalizer slice that forces a trace-zero witness whenever an alternating row-prefix invariant, equivalently a signed \(a\)-count, is nonzero.
\item \textbf{The quaternionic criterion:} an algebraic specialization on the principal slope-visible family that captures nonvanishing metabelian data through the rigid symmetries of the quaternion group.
\item \textbf{The local one-row criterion:} a geometric construction that isolates a local commutator pattern when the discrepancy between the two words is confined to a single contiguous block move.
\end{itemize}
Two supplementary mechanisms further enlarge the range: a binary-dihedral symmetry test that reads a word exponent \(\kappa(w)\), and an interior-point filter that exploits explicit non-positive evaluations of the Fricke--Vogt trace polynomial in trace coordinates. In experiments on \(50{,}000\) randomly generated hard positive-word pairs, this integrated menu found a trace-zero witness for every pair.

Formalizing these criteria also clarifies their present boundary. The constructions isolate a sharply delimited residual class, termed \emph{super-degenerate}, which presently marks the geometric frontier of the two-state exact-separation conjecture. The resulting perspective reduces the conjecture from an unbounded algebraic search to a focused extension problem on this residual locus.

The paper is organized as follows. Section~\ref{sec:preliminaries} develops the metabelian framework and the slope-specialization tools used throughout. Section~\ref{sec:witnesses} constructs the main slice families and extracts the certified witness mechanisms. Section~\ref{sec:obstructions} synthesizes these criteria, isolates the super-degenerate residual class, and proves the finite-image obstruction. Section~\ref{sec:conclusion} concludes with a brief outlook.

\section{Algebraic Foundations and Metabelian Invariants}\label{sec:preliminaries}

This section establishes the algebraic and automata-theoretic foundations used throughout the paper. The material is organized in three parts: basic free-group conventions and the abelianization map; the Fox free differential calculus and the metabelian polynomial, including the explicit interval-block formula for positive-word differences; and the trace-zero witness formulation that connects these algebraic objects to the two-state QFA separation problem of~\cite{BMY16}.

\subsection{Basic notations}
Let \(\F=\langle a,b\rangle\) denote the free group on two generators~\cite{LyndonSchupp77,Roman12}. A \emph{positive word} in \(a\) and \(b\) means a finite sequence of letters from \(\{a,b\}\). The set of all such words, including the empty word \(1\), is denoted by \(\{a,b\}^{\ast}\). Via the natural embedding \(\{a,b\}^{\ast}\hookrightarrow \F\), every positive word is regarded as an element of \(\F\). For \(x\in \{a,b\}^{\ast}\), write \(\#_a(x)\) and \(\#_b(x)\) for the numbers of occurrences of \(a\) and \(b\) in \(x\). In particular,
\[
\#_a(1)=\#_b(1)=0.
\]

For elements \(x,y\) of a group, write
\[
[x,y]:=x^{-1}y^{-1}xy.
\]
For subgroups \(H,K\le G\), write
\[
[H,K]:=\angles{[h,k]:h\in H,\ k\in K}.
\]
In particular,
\[
\F'=[\F,\F],
\qquad
\F''=[\F',\F'].
\]
The derived series is defined by \(\F^{(1)}=\F'\) and \(\F^{(2)}=\F''\), and the lower central series by \(\gamma_1(\F)=\F\) and \(\gamma_{k+1}(\F)=[\gamma_k(\F),\F]\).

If \(X\) and \(Y\) are sets, then \(X\times Y\) denotes their Cartesian product. More generally, \(X^2:=X\times X\). The abelianization map
\[
\ab:\F\to \Z^2
\]
is the group homomorphism from \(\F\) to the additive group \(\Z^2\) determined by
\[
\ab(a)=(1,0),
\qquad
\ab(b)=(0,1),
\qquad
\ab(a^{-1})=(-1,0),
\qquad
\ab(b^{-1})=(0,-1).
\]
Equivalently, \(\ab(xy)=\ab(x)+\ab(y)\) for all \(x,y\in \F\). Thus, if positive words \(u\) and \(v\) satisfy \(\ab(u)=\ab(v)\), then a direct calculation gives
\[
\ab(u^{-1}v)=-\ab(u)+\ab(v)=0,
\]
so
\[
w:=u^{-1}v\in \ker(\ab)=\F'.
\]

If \(N\triangleleft G\) is a normal subgroup, then \(G/N\) denotes the quotient group of cosets \(gN\), with multiplication \((gN)(hN)=(gh)N\). In particular, \(\F/\F''\) is the metabelian quotient of \(\F\).

If \(R\) is a commutative ring, then \(R[x]\) denotes the polynomial ring in the indeterminate \(x\).

\subsection{Fox calculus and the metabelian polynomial}

Work in the integral group ring \(\Z[\F]\). The free differential calculus of Fox~\cite{Fox53} provides derivations
\[
\frac{\partial}{\partial a},\frac{\partial}{\partial b}:\Z[\F]\to \Z[\F]
\]
determined by
\[
\frac{\partial(uv)}{\partial b}=\frac{\partial u}{\partial b}+u\frac{\partial v}{\partial b},
\qquad
\frac{\partial a}{\partial b}=0,
\qquad
\frac{\partial b}{\partial b}=1,
\qquad
\frac{\partial(x^{-1})}{\partial b}=-x^{-1}\frac{\partial x}{\partial b}.
\]
Composing with the abelianization ring homomorphism
\[
\pi_{\mathrm{ab}}:\Z[\F]\to \laurent,
\]
determined by \(a\mapsto T\) and \(b\mapsto S\), produces Laurent polynomials that encode the metabelian information carried by the Fox derivative. For \(f\in \Z[\F]\), write \(f^{\mathrm{ab}}:=\pi_{\mathrm{ab}}(f)\). This avoids any clash with complex conjugation, which appears later only for scalar-valued functions on the unit circle.

\begin{proposition}[cf.~\cite{Fox53,LyndonSchupp77}]\label{prop:metabelian-polynomial}
Let \(\pi_{\mathrm{ab}}:\Z[\F]\to \laurent\) be the abelianization ring homomorphism defined above. For \(w\in\F\), define
\[
B_w(T,S):=\left(\frac{\partial w}{\partial b}\right)^{\mathrm{ab}}\in \laurent.
\]
If \(w\in\F'\), then \(B_w(1,S)=0\), so there exists a unique Laurent polynomial \(M_w(T,S)\in \laurent\) such that
\[
B_w(T,S)=-(T-1)\,M_w(T,S).
\]
The Laurent polynomial \(M_w(T,S)\) is called the \emph{metabelian polynomial} of \(w\). Moreover,
\[
M_w(T,S)=0 \iff B_w(T,S)=0 \iff w\in \F''.
\]
In particular, \(M_w\neq 0\) implies that \(w\) remains nontrivial in \(\F/\F''\).
\end{proposition}

For positive words \(u\) and \(v\) with the same letter counts, write \(\ab(u)=\ab(v)=(m,n)\) and \(w=u^{-1}v\). For each row \(j\in\{1,\dots,n\}\), let \(A_u(j)\) be the number of \(a\)'s that appear before the \(j\)-th \(b\) in \(u\), and define \(A_v(j)\) likewise. Set
\[
\delta_j:=A_v(j)-A_u(j),\qquad \alpha_j:=\min\{A_u(j),A_v(j)\},\qquad \eta_j:=\operatorname{sgn}(\delta_j),
\]
and let \(\Jset=\{j\in\{1,\dots,n\}:\delta_j\neq 0\}\). The formula below is the structural algebraic statement used throughout the remainder of the paper.

\begin{theorem}[cf.~\cite{LyndonSchupp77}]\label{thm:rowwise-structure}
Let \(u \neq v \in \{a,b\}^\ast\) be positive words with \(\ab(u)=\ab(v)=(m,n)\), and let \(w=u^{-1}v \in \F'\). With the notation above,
\begin{equation}\label{eq:rowwise-interval-block}
M_w(T,S) = T^{-m}S^{-n}\sum_{j\in \Jset}(-\eta_j)\,T^{\alpha_j}\,S^{j-1}\,\bigl(1+T+\cdots+T^{|\delta_j|-1}\bigr).
\end{equation}
\end{theorem}

Formula~\eqref{eq:rowwise-interval-block} captures the metabelian footprint of a positive-word difference and provides the sole rowwise algebraic input for the framework developed later. The rest of this section prepares the passage from this two-variable invariant to \(SU(2)\) trace geometry. In particular, we show that after a suitable integer-slope compression the metabelian data still retain nontrivial one-variable information, which will become the leading trace-deficit term on the explicit slice families constructed in Section~\ref{sec:witnesses}.

\subsection{Metabelian persistence and slope specializations}\label{subsec:metabelian}

Before the metabelian polynomial \(M_w(T,S)\) can drive the trace geometry on a one-parameter slice, it must survive the passage from two variables to one. This requires two algebraic facts: first, \(M_w(T,S)\) is nonzero for every hard pair; second, there exists an integer slope specialization \(S=T^r\) that preserves this nontriviality in a form usable later for the trace-deficit analysis.

For an integer slope $r \in \Z$, define the one-variable specialization
\[
p_{w,r}(T) := M_w(T,T^r) \in \Z[T^{\pm1}],
\]
and the associated Laurent polynomial
\[
P_r(T) := (T-1)M_w(T,T^r) = (T-1)p_{w,r}(T) \in \Z[T^{\pm1}].
\]

Before choosing a slope, it is necessary to know that the original two-variable polynomial $M_w(T,S)$ is not identically zero.

\begin{lemma}\label{lem:metabelian-nondeg}
Let $u \neq v \in \{a,b\}^\ast$ be distinct positive words with $\ab(u)=\ab(v)$, and set $w=u^{-1}v\in\F'$. Then $w\notin\F''$. Equivalently, $u$ and $v$ remain distinct in the metabelian quotient $\F/\F''$.
\end{lemma}

\begin{proof}
Write $\ab(u)=\ab(v)=(m,n)$. The case \(n=0\) cannot occur under the assumptions, since then both positive words with abelianization \((m,0)\) are equal to \(a^m\). Thus \(n\ge 1\). For each $j\in\{1,\dots,n\}$, let $A_u(j)$ and $A_v(j)$ be the prefix counts defined above. A positive word with abelianization $(m,n)$ is uniquely determined by the sequence $(A_x(j))_{j=1}^n$, since one can reconstruct it as
\[
x=a^{A_x(1)} b\, a^{A_x(2)-A_x(1)} b\cdots b\, a^{m-A_x(n)}.
\]
Therefore $u\neq v$ implies that $A_u(j)\neq A_v(j)$ for at least one $j$, hence $\Jset\neq\varnothing$.

Now apply Theorem~\ref{thm:rowwise-structure}. For each $j\in \Jset$, the corresponding contribution to $M_w(T,S)$ is
\[
(-\eta_j)\,T^{-m+\alpha_j}S^{j-1-n}\bigl(1+T+\cdots+T^{|\delta_j|-1}\bigr),
\]
which is a nonzero Laurent polynomial. Since distinct indices $j$ carry distinct powers of $S$, no cancellation can occur between different $j$. Because $\Jset\neq\varnothing$, formula~\eqref{eq:rowwise-interval-block} shows that $M_w(T,S)\neq 0$.

By Proposition~\ref{prop:metabelian-polynomial}, the nonvanishing of $M_w(T,S)$ is equivalent to $w\notin \F''$. This proves the claim.
\end{proof}

Once $M_w(T,S)$ is known to be nonzero, the next step is to choose a specialization that keeps it nonzero.

\begin{lemma}\label{lem:slope-specialization}
Let $f(T,S)\in \mathbb Z[T^{\pm1},S^{\pm1}]$ be a nonzero Laurent polynomial with finite support $\Sigma\subset \mathbb Z^2$:
\[
f(T,S)=\sum_{(i,j)\in \Sigma} c_{i,j}\,T^iS^j.
\]
Then there exists an integer $r$ such that the specialization $f(T,T^r)$ is not identically zero.
\end{lemma}

The only obstruction comes from finitely many slopes that force collisions among distinct support exponents.
\begin{proof}
When we substitute $S=T^r$, the term $T^iS^j$ becomes $T^{i+rj}$. Cancellation can occur only if two distinct support points $(i,j)$ and $(i',j')$ are mapped to the same exponent:
\[
i+rj=i'+rj'
\quad\Longleftrightarrow\quad
r=\frac{i-i'}{j'-j}.
\]
For each pair of distinct support points with $j\neq j'$, this determines at most one rational bad slope. Since $\Sigma$ is finite, only finitely many bad slopes arise. Choose an integer $r$ outside that finite set. Then the map $(i,j)\mapsto i+rj$ is injective on $\Sigma$, so no two support monomials collide, and therefore $f(T,T^r)\not\equiv 0$.
\end{proof}

Applying the specialization lemma to the metabelian polynomial gives the following immediate corollary.

\begin{corollary}\label{cor:metabelian-nonvanishing-specialization}
If $w\in \F'\setminus \F''$, there exists an integer $r$ such that the specialized polynomial $p_{w,r}(T)=M_w(T,T^r)$ is not identically zero.
\end{corollary}

\begin{proof}
This is a direct application of Lemma~\ref{lem:slope-specialization}. Since $w\in \F'\setminus \F''$, its metabelian polynomial $M_w(T,S)$ is a nonzero Laurent polynomial. Substituting $f(T,S)=M_w(T,S)$ into Lemma~\ref{lem:slope-specialization} yields an integer slope $r$ for which $p_{w,r}(T)\neq 0$.
\end{proof}

The slope specialization established above guarantees nonvanishing only at the level of Laurent polynomials. For exact separation, however, the relevant object is an evaluation on explicit unitary matrices. Before constructing the geometric slice families in Section~\ref{sec:witnesses} that realize this passage, we first formulate the matrix-theoretic target itself, namely the notion of a trace-zero witness.

\subsection{Trace-zero witnesses}

Finally, we connect these algebraic and calculus frameworks to matrices. For any word \(w\in F_2\), a \emph{trace-zero witness} is a pair of \(2 \times 2\) unitary matrices \((A,B)\in SU(2)^2\) such that the trace of the evaluated word \(w(A,B)\), obtained by substituting \(A\) for \(a\) and \(B\) for \(b\), is exactly zero:
\[
\tr\bigl(w(A,B)\bigr)=0.
\]
When \(u,v\in\{a,b\}^\ast\) satisfy \(\ab(u)=\ab(v)\), the difference word \(w:=u^{-1}v\) lies in \(\F'\). A central problem is to determine when such a word admits a trace-zero witness, and in particular how the metabelian polynomial controls that question for positive-word differences. In the present paper, that control is pursued through explicit slice families in \(SU(2)^2\): the algebraic data are used to choose or analyze low-dimensional families on which the trace deficit can be computed effectively.

\begin{remark}
In the standard two-state \(SU(2)\) model for measure-once quantum finite automata on \(\{a,b\}\), a trace-zero witness yields exact separation for the pair \((u,v)\), as explained in~\cite{BMY16}. The remainder of the paper develops such witnesses by a slice-driven route: first a principal slope-visible family is constructed, and then additional concrete slice families and specializations are extracted from it or placed alongside it.
\end{remark}

\section{Explicit slice families and trace-zero witnesses}\label{sec:witnesses}

Building on the algebraic foundations of Section~\ref{sec:preliminaries}, this section activates the analytic engine of the framework. By placing the specialized Fox polynomial on a principal slope-visible slice, we prove a quadratic trace-deficit identity that drives evaluations toward the trace-zero locus. From this geometric core, we then extract a certified menu of explicit witnesses, coming from global dihedral and quaternionic specializations together with local row-wise configurations and supplementary explicit tests. Organized together, these constructions form a layered certified menu and provide an effective separation strategy for a broad range of hard word pairs.

\subsection{Metabelian witnesses on the slope-visible slice}

We now localize the algebraic persistence established in Subsection~\ref{subsec:metabelian}. Expanding the word map near a commuting basepoint on the principal slope-visible family shows that the specialized Fox polynomial governs the leading trace behavior in a precise way. The main result of this subsection is the quadratic identity that makes this control explicit. We then evaluate distinguished specializations the same family to obtain the dihedral and quaternionic witnesses used later in the paper.

\subsubsection{An explicit one-parameter family and the quadratic Fox term}\label{subsec:fox-quadratic}

The results of Subsection~\ref{subsec:metabelian} are purely algebraic. To connect them to actual trace-zero witnesses, a direct link to \(SU(2)\) evaluations is needed. The one-parameter family \((A(\theta),B_r(\theta,t))\) defined below provides that link: it turns the specialized Fox polynomial into the leading term of the trace deficit near a commuting point, so that the nontriviality established above begins to control the geometry of the witness problem.

Fix an integer \(r\in\Z\). Let \(q=e^{i\theta}\), set \(T=q^2=e^{i2\theta}\), and define
\[
A(\theta)=\begin{pmatrix}q&0\\0&q^{-1}\end{pmatrix},\qquad
D_r(\theta)=\begin{pmatrix}q^{r}&0\\0&q^{-r}\end{pmatrix}.
\]
Let
\[
J_0:=\begin{pmatrix}0&1\\-1&0\end{pmatrix},
\qquad
R(t):=\begin{pmatrix}\cos t & \sin t\\ -\sin t & \cos t\end{pmatrix}=\exp(tJ_0)\in SU(2).
\]
The matrix \(J_0\) will reappear several times; it satisfies \(J_0^2=-I\) and \(R(\pi/2)=J_0\). Then set
\[
B_r(\theta,t):=D_r(\theta)\,R(t)\in SU(2).
\]
For \(w\in \F'\), define
\[
U_{w,r}(\theta,t):=w\bigl(A(\theta),B_r(\theta,t)\bigr)\in SU(2),
\qquad
f_{w,r}(T,t):=\tr\bigl(U_{w,r}(\theta,t)\bigr)\in[-2,2].
\]
Because \(w\in \F'\) has trivial abelianization, replacing \(q\) by \(-q\) multiplies \(A(\theta)\) and \(B_r(\theta,t)\) only by central signs, and these signs cancel in the evaluation of \(w\). Hence \(U_{w,r}(\theta,t)\), and therefore its trace, depends only on \(T=q^2\). The notation \(f_{w,r}(T,t)\) is therefore well defined.
The slice is chosen so that the point \(t=0\) is commuting: indeed \(B_r(\theta,0)=D_r(\theta)\) is diagonal and therefore commutes with \(A(\theta)\). Since \(w\in F_2'\) has trivial abelianization, evaluating \(w\) on the commuting pair \((A(\theta),D_r(\theta))\) gives the identity matrix. Hence
\[
U_{w,r}(\theta,0)=I,
\qquad
f_{w,r}(T,0)=2.
\]
The perturbation parameter \(t\) therefore measures how the trace moves away from this commuting reference point.

The first basic property of the slice is that the trace depends on \(t\) through an even real-analytic function.

\begin{lemma}\label{lem:trace-even-perturb}
For each fixed \(\theta\) (equivalently fixed \(T=q^2\)), the function \(t\mapsto f_{w,r}(T,t)\) is real-analytic and even.
Consequently it admits a Taylor expansion
\begin{equation}\label{eq:trace-even-expansion}
f_{w,r}(T,t)=2-\sum_{m\ge 1} C_{2m}(T)\,t^{2m},
\end{equation}
where \(C_{2m}(T)\in \mathbb Q[T^{\pm1}]\). The expansion therefore starts at order \(t^2\).
\end{lemma}

\begin{proof}
Real-analyticity in \(t\) is immediate because the entries of \(R(t)\) are \(\cos t\) and \(\sin t\), and the word map is obtained from matrix multiplication and inversion, both of which are analytic in the matrix entries.

For evenness, set \(K:=\mathrm{diag}(i,-i)\in SU(2)\). A direct calculation gives \(K R(t)K^{-1}=R(-t)\). Since \(K\) is diagonal and therefore commutes with both \(A(\theta)\) and \(D_r(\theta)\), it follows that
\[
K B_r(\theta,t)K^{-1}=B_r(\theta,-t),
\qquad
K A(\theta)K^{-1}=A(\theta).
\]
Therefore
\[
U_{w,r}(\theta,-t)
=w\bigl(KA(\theta)K^{-1},\,K B_r(\theta,t)K^{-1}\bigr)
=K\,U_{w,r}(\theta,t)\,K^{-1},
\]
and taking traces yields \(f_{w,r}(T,-t)=f_{w,r}(T,t)\).

Finally, write all entries as functions of \(q=e^{i\theta}\). Since \(R(t)=\exp(tJ_0)\) has Taylor coefficients in \(\mathbb Q\), each Taylor coefficient of the matrix \(U_{w,r}(\theta,t)\) is a finite Laurent polynomial in \(q\) with rational coefficients. Because \(w\in \F'\), replacing \(q\) by \(-q\) multiplies \(A(\theta)\) and \(B_r(\theta,t)\) by central signs whose total contribution in the word evaluation is trivial. Hence \(U_{w,r}(\theta,t)\), and therefore its trace, is invariant under \(q\mapsto -q\). Each trace coefficient is therefore an even Laurent polynomial in \(q\), equivalently a Laurent polynomial in \(T=q^2\). Thus \(C_{2m}(T)\in\mathbb Q[T^{\pm1}]\).
\end{proof}

The next proposition is the precise bridge from the slice to Fox calculus: the linear term of the off-diagonal entry is the specialized abelianized Fox derivative, and the quadratic trace drop is its squared modulus.

\begin{proposition}\label{prop:trace-quadratic-fox}
Write
\[
U_{w,r}(\theta,t)=
\begin{pmatrix}
\alpha_{w,r}(T,t) & \beta_{w,r}(T,t)\\
-\overline{\beta_{w,r}(T,t)} & \overline{\alpha_{w,r}(T,t)}
\end{pmatrix}\in SU(2),
\qquad T=q^2=e^{i2\theta}.
\]
Then
\[
\alpha_{w,r}(T,0)=1,
\qquad
\beta_{w,r}(T,0)=0.
\]
If \(S:=T^r=q^{2r}\), then the first derivative at the commuting point \(t=0\) is
\begin{equation}\label{eq:beta-prime-fox}
\partial_t\beta_{w,r}(T,0)
= S\,B_w(T,S)\big|_{S=T^r}
= T^r\,B_w(T,T^r)
= -T^r(T-1)\,M_w(T,T^r).
\end{equation}
Moreover, the trace expansion \eqref{eq:trace-even-expansion} satisfies
\begin{equation}\label{eq:C2-is-fox-square}
C_2(T)=\bigl|\partial_t\beta_{w,r}(T,0)\bigr|^2
=\bigl|B_w(T,T^r)\bigr|^2
=\bigl|(T-1)M_w(T,T^r)\bigr|^2.
\end{equation}
Since \(T=e^{i2\theta}\) lies on the unit circle, for every Laurent polynomial \(P\in\mathbb R[T^{\pm1}]\) one has \(|P(T)|^2=P(T)P(T^{-1})\). Thus the right-hand side is again a Laurent polynomial in \(T\).
\end{proposition}

\begin{proof}
As noted above, \(w\in F_2'\) and the pair \((A(\theta),D_r(\theta))\) is commuting. Therefore
\[
U_{w,r}(\theta,0)=w\bigl(A(\theta),D_r(\theta)\bigr)=I,
\]
so \(\alpha_{w,r}(T,0)=1\) and \(\beta_{w,r}(T,0)=0\).

To compute the first derivative, note that \(B_r(\theta,t)=D_r(\theta)e^{tJ_0}\), hence
\[
\partial_t B_r(\theta,t)\big|_{t=0}=D_r(\theta)J_0,
\qquad
\partial_t B_r(\theta,t)^{-1}\big|_{t=0}=-J_0 D_r(\theta)^{-1}.
\]
Choose a reduced expression \(w=x_1\cdots x_N\) with \(x_k\in\{a^{\pm1},b^{\pm1}\}\). For each \(k\), let
\[
P_k:=x_1\cdots x_{k-1},
\qquad
Q_k:=x_{k+1}\cdots x_N,
\]
so that \(w=P_kx_kQ_k\).

Differentiate the product
\[
x_1\bigl(A(\theta),B_r(\theta,t)\bigr)\cdots x_N\bigl(A(\theta),B_r(\theta,t)\bigr)
\]
at \(t=0\). Only those factors with \(x_k=b^{\pm1}\) contribute. If \(x_k=b\), the corresponding term is
\[
\begin{aligned}
&P_k\bigl(A(\theta),D_r(\theta)\bigr)\,\partial_t B_r(\theta,t)\big|_{t=0}\,
Q_k\bigl(A(\theta),D_r(\theta)\bigr)\\
&\qquad=P_k\bigl(A(\theta),D_r(\theta)\bigr)D_r(\theta)J_0Q_k\bigl(A(\theta),D_r(\theta)\bigr).
\end{aligned}
\]
Since
\[
P_k\bigl(A(\theta),D_r(\theta)\bigr)D_r(\theta)Q_k\bigl(A(\theta),D_r(\theta)\bigr)
=
U_{w,r}(\theta,0)=I,
\]
this becomes
\[
\Bigl(P_k\bigl(A(\theta),D_r(\theta)\bigr)D_r(\theta)\Bigr)
J_0
\Bigl(P_k\bigl(A(\theta),D_r(\theta)\bigr)D_r(\theta)\Bigr)^{-1}.
\]
If instead \(x_k=b^{-1}\), then the corresponding term is
\[
\begin{aligned}
&P_k\bigl(A(\theta),D_r(\theta)\bigr)\,\partial_t B_r(\theta,t)^{-1}\big|_{t=0}\,
Q_k\bigl(A(\theta),D_r(\theta)\bigr)\\
&\qquad=- P_k\bigl(A(\theta),D_r(\theta)\bigr)J_0 D_r(\theta)^{-1}Q_k\bigl(A(\theta),D_r(\theta)\bigr).
\end{aligned}
\]
Now
\[
P_k\bigl(A(\theta),D_r(\theta)\bigr)D_r(\theta)^{-1}Q_k\bigl(A(\theta),D_r(\theta)\bigr)
=
U_{w,r}(\theta,0)=I,
\]
so this contribution simplifies to
\[
- P_k\bigl(A(\theta),D_r(\theta)\bigr)
J_0
P_k\bigl(A(\theta),D_r(\theta)\bigr)^{-1}.
\]
Thus every first-order term is off-diagonal, and the upper-right entry can be read off by conjugating \(J_0\) by diagonal prefixes.

Write \(\ab(P_k)=(m_k,n_k)\). Because \(A(\theta)\) and \(D_r(\theta)\) are diagonal,
\[
P_k\bigl(A(\theta),D_r(\theta)\bigr)=\operatorname{diag}(z_k,z_k^{-1}),
\qquad
z_k^2=T^{m_k}S^{n_k},
\]
and similarly
\[
P_k\bigl(A(\theta),D_r(\theta)\bigr)D_r(\theta)=\operatorname{diag}(z_k',z_k'^{-1}),
\qquad
(z_k')^2=T^{m_k}S^{n_k+1}.
\]
For any nonzero complex number \(z\),
\[
\operatorname{diag}(z,z^{-1})J_0\operatorname{diag}(z,z^{-1})^{-1}
=
\begin{pmatrix}0&z^2\\-z^{-2}&0\end{pmatrix}.
\]
Hence the upper-right entry of \(\partial_tU_{w,r}(\theta,t)|_{t=0}\) is
\[
\begin{aligned}
\partial_t\beta_{w,r}(T,0)
&=\sum_{x_k=b} T^{m_k}S^{n_k+1}
-\sum_{x_k=b^{-1}} T^{m_k}S^{n_k}\\
&=S\left(
\sum_{x_k=b} T^{m_k}S^{n_k}
-\sum_{x_k=b^{-1}} T^{m_k}S^{n_k-1}
\right).
\end{aligned}
\]
Now apply the Fox rules from Section~\ref{sec:preliminaries}. Repeatedly using
\[
\frac{\partial(uv)}{\partial b}=\frac{\partial u}{\partial b}+u\frac{\partial v}{\partial b},
\qquad
\frac{\partial b}{\partial b}=1,
\qquad
\frac{\partial(b^{-1})}{\partial b}=-b^{-1},
\]
one obtains the standard reduced-word formula
\[
\frac{\partial w}{\partial b}
=
\sum_{x_k=b} P_k
-
\sum_{x_k=b^{-1}} P_k b^{-1}.
\]
After abelianization, this becomes
\[
B_w(T,S)
=
\sum_{x_k=b} T^{m_k}S^{n_k}
-
\sum_{x_k=b^{-1}} T^{m_k}S^{n_k-1}.
\]
Comparing with the previous display yields
\[
\partial_t\beta_{w,r}(T,0)
= S\,B_w(T,S)\big|_{S=T^r}
= T^r\,B_w(T,T^r).
\]
Since \(w\in \F'\), Proposition~\ref{prop:metabelian-polynomial} gives
\[
B_w(T,S)=-(T-1)M_w(T,S),
\]
which proves \eqref{eq:beta-prime-fox}.

For the quadratic trace coefficient, use the identity
\[
2-\mathrm{tr}(U)=|\alpha-1|^2+|\beta|^2,
\qquad
U=\begin{pmatrix}\alpha&\beta\\-\overline\beta&\overline\alpha\end{pmatrix}\in SU(2).
\]
The first-order matrix computed above is purely off-diagonal, so \(\partial_t\alpha_{w,r}(T,0)=0\). Consequently,
\[
\alpha_{w,r}(T,t)-1=O(t^2),
\qquad
\beta_{w,r}(T,t)=\partial_t\beta_{w,r}(T,0)\,t+O(t^2).
\]
Therefore
\[
|\alpha_{w,r}(T,t)-1|^2=O(t^4),
\qquad
|\beta_{w,r}(T,t)|^2=\bigl|\partial_t\beta_{w,r}(T,0)\bigr|^2 t^2+O(t^3),
\]
and hence
\[
2-\mathrm{tr}\bigl(U_{w,r}(\theta,t)\bigr)
=\bigl|\partial_t\beta_{w,r}(T,0)\bigr|^2 t^2+O(t^3).
\]
By Lemma~\ref{lem:trace-even-perturb}, the left-hand side is an even function of \(t\), so the remainder term is in fact \(O(t^4)\). Comparing with \eqref{eq:trace-even-expansion} yields
\[
C_2(T)=\bigl|\partial_t\beta_{w,r}(T,0)\bigr|^2.
\]
Finally, because \(T=e^{i2\theta}\) lies on the unit circle, one has \(|T^r|=1\). Taking absolute values in \eqref{eq:beta-prime-fox} therefore gives
\[
\bigl|\partial_t\beta_{w,r}(T,0)\bigr|^2
=\bigl|B_w(T,T^r)\bigr|^2
=\bigl|(T-1)M_w(T,T^r)\bigr|^2,
\]
which is exactly \eqref{eq:C2-is-fox-square}.
\end{proof}

This explicit quadratic control is the operational payoff of the principal slice. By Equation~\eqref{eq:C2-is-fox-square}, one does not need to analyze the full noncommutative product \(w(A(\theta),B_r(\theta,t))\) near the commuting point \(t=0\): the specialized Fox polynomial alone determines whether the trace has a nonzero quadratic deficit away from \(2\), and it determines the leading size of that deficit. More precisely, the trace begins to drop from \(2\) to second order exactly when \((T-1)M_w(T,T^r)\neq 0\), and the corresponding quadratic deficit is precisely \(\bigl|(T-1)M_w(T,T^r)\bigr|^2\). This is the mechanism used repeatedly in the remaining witness constructions.

\begin{remark}
Beyond this local quadratic deficit, one can also evaluate \(w\) on a Laurent-matrix algebra over \(\Z[T^{\pm1}]\) and obtain an exact global sum-of-two-squares identity for the resulting trace deficit. That complementary identity is not used in the discrete criteria developed below, but it points in the same direction as the completion target formulated later in Conjecture~\ref{conj:finite-menu-cover}.
\end{remark}

\subsubsection{Explicit specializations of the one-parameter family}\label{subsec:algebraic-points}

Keep the notation of Subsubsection~\ref{subsec:fox-quadratic}. Thus
\[
\begin{aligned}
A(\theta)&=\mathrm{diag}(q,q^{-1}),\qquad D_r(\theta)=\mathrm{diag}(q^r,q^{-r}),\\
B_r(\theta,t)&=D_r(\theta)R(t),\qquad T=q^2.
\end{aligned}
\]
Two special evaluations will be used repeatedly. The choice \(t=\pi/2\) gives the matrix
\[
J_0=R(\pi/2)=\begin{pmatrix}0&1\\-1&0\end{pmatrix},
\]
the same matrix already introduced in Subsubsection~\ref{subsec:fox-quadratic}, which conjugates diagonal matrices to their inverses. The choice \(q=i\), equivalently \(T=-1\), gives the diagonal element \(\mathrm{diag}(i,-i)\). These are the two specializations used in the proofs below.

For the dihedral specialization \(t=\pi/2\), write
\[
B_r(\theta):=B_r(\theta,\pi/2)=D_r(\theta)J_0.
\]

\begin{proposition}\label{prop:dihedral-point-trace}
Define
\begin{equation}\label{eq:Delta0-rowprefix}
\Delta_0(w):=\sum_{j=1}^{n}(-1)^{j-1}\bigl(A_v(j)-A_u(j)\bigr)\in\Z.
\end{equation}
Then for every \(\theta\), equivalently every \(T=q^2\) on the unit circle,
\begin{equation}\label{eq:trace-dihedral-point}
\begin{aligned}
w\bigl(A(\theta),B_r(\theta)\bigr)&=A(\theta)^{2(-1)^n\Delta_0(w)},\\
\tr\bigl(w(A(\theta),B_r(\theta))\bigr)&=T^{\Delta_0(w)}+T^{-\Delta_0(w)}.
\end{aligned}
\end{equation}
In particular, \(\Delta_0(w)\) and \eqref{eq:trace-dihedral-point} are independent of the integer \(r\). If \(\Delta_0(w)\neq 0\), then choosing
\[
\theta=\frac{\pi}{4|\Delta_0(w)|}
\]
produces a trace-zero witness on this slice.
\end{proposition}

\begin{proof}
Recall that \(J_0=R(\pi/2)=\begin{psmallmatrix}0&1\\-1&0\end{psmallmatrix}\). Then \(J_0^2=-I\) and
\[
J_0 A(\theta)^k J_0^{-1}=A(\theta)^{-k}\qquad (k\in\Z).
\]
Since \(D_r(\theta)\) commutes with \(A(\theta)\) and satisfies \(J_0 D_r(\theta)J_0^{-1}=D_r(\theta)^{-1}\), the specialized matrix
\[
B_r(\theta)=D_r(\theta)J_0
\]
obeys
\[
B_r(\theta)A(\theta)^k B_r(\theta)^{-1}=A(\theta)^{-k},
\qquad
B_r(\theta)^2=-I.
\]

Write a positive word \(x\in\{a,b\}^{\ast}\) with \(n=\#_b(x)\) in the form
\[
x=a^{x_0}ba^{x_1}b\cdots ba^{x_n},\qquad x_t\ge 0.
\]
Repeatedly commuting \(B_r(\theta)\) past powers of \(A(\theta)\) gives
\[
x\bigl(A(\theta),B_r(\theta)\bigr)=A(\theta)^{\nu(x)}B_r(\theta)^n,
\qquad
\nu(x):=\sum_{t=0}^{n}(-1)^t x_t.
\]
Let \(A_x(j)=x_0+\cdots+x_{j-1}\). Since \(x_{j-1}=A_x(j)-A_x(j-1)\) for \(j\ge 2\), a telescoping sum gives
\[
\nu(x)=2\sum_{j=1}^{n}(-1)^{j-1}A_x(j)+(-1)^n m,
\]
where \(m=\#_a(x)\).

Apply this to \(u\) and \(v\). Because \(\ab(u)=\ab(v)=(m,n)\),
\[
u\bigl(A(\theta),B_r(\theta)\bigr)=A(\theta)^{\nu(u)}B_r(\theta)^n,
\qquad
v\bigl(A(\theta),B_r(\theta)\bigr)=A(\theta)^{\nu(v)}B_r(\theta)^n.
\]
Hence
\[
\begin{aligned}
w\bigl(A(\theta),B_r(\theta)\bigr)
&=u\bigl(A(\theta),B_r(\theta)\bigr)^{-1}v\bigl(A(\theta),B_r(\theta)\bigr)\\
&=B_r(\theta)^{-n}A(\theta)^{\nu(v)-\nu(u)}B_r(\theta)^n\\
&=A(\theta)^{(-1)^n(\nu(v)-\nu(u))}.
\end{aligned}
\]
The \((-1)^n m\) terms cancel in \(\nu(v)-\nu(u)\), so
\[
(-1)^n(\nu(v)-\nu(u))
=2(-1)^n\sum_{j=1}^{n}(-1)^{j-1}\bigl(A_v(j)-A_u(j)\bigr)
=2(-1)^n\Delta_0(w).
\]
This proves the first identity in \eqref{eq:trace-dihedral-point}. Taking traces gives
\[
\tr\bigl(w(A(\theta),B_r(\theta))\bigr)
=q^{2(-1)^n\Delta_0(w)}+q^{-2(-1)^n\Delta_0(w)}
=T^{\Delta_0(w)}+T^{-\Delta_0(w)}.
\]
If \(\Delta_0(w)\neq 0\), choose \(\theta=\pi/(4|\Delta_0(w)|)\). Then \(T=e^{i\pi/(2|\Delta_0(w)|)}\), so
\[
T^{\Delta_0(w)}+T^{-\Delta_0(w)}
=e^{i\,\mathrm{sgn}(\Delta_0(w))\pi/2}+e^{-i\,\mathrm{sgn}(\Delta_0(w))\pi/2}=0.
\]
Only the relations \(B_r(\theta)A(\theta)B_r(\theta)^{-1}=A(\theta)^{-1}\) and \(B_r(\theta)^2=-I\) were used, so the formula is independent of \(r\).
\end{proof}

The term \emph{dihedral point} refers to the subgroup generated by \(A(\theta)\) and \(B_r(\theta)=D_r(\theta)J_0\). The conjugation relation established in the proof, together with \(B_r(\theta)^2=-I\), shows that modulo the center \(\{\pm I\}\) this subgroup is dihedral, with \(A(\theta)\) playing the rotation and \(B_r(\theta)\) the reflection.

\begin{corollary}\label{cor:mu-delta0-equivalence}
Let
\[
u=a^{u_0}ba^{u_1}b\cdots ba^{u_n},
\qquad
v=a^{v_0}ba^{v_1}b\cdots ba^{v_n},
\]
with \(\ab(u)=\ab(v)\), and set \(w:=u^{-1}v\). Define
\[
\mu(u):=\sum_{t=0}^{n}(-1)^t u_t,
\qquad
\mu(v):=\sum_{t=0}^{n}(-1)^t v_t.
\]
Then
\[
\mu(v)-\mu(u)=2\Delta_0(w).
\]
In particular, \(\Delta_0(w)\neq 0\) if and only if \(\mu(u)\neq\mu(v)\). Moreover, the mixed normalizer slice \((A(\theta),J_0)\) is the \(r=0\) member of the dihedral family from Proposition~\ref{prop:dihedral-point-trace}.
\end{corollary}

\begin{proof}
The calculation in Proposition~\ref{prop:dihedral-point-trace} shows that for any positive word
\[
x=a^{x_0}ba^{x_1}b\cdots ba^{x_n},
\]
the exponent occurring in \(x(A(\theta),B_r(\theta))\) is
\[
\nu(x)=\sum_{t=0}^{n}(-1)^t x_t.
\]
Applying this to \(u\) and \(v\) gives \(\mu(u)=\nu(u)\) and \(\mu(v)=\nu(v)\), and the same computation yields
\[
(-1)^n\bigl(\mu(v)-\mu(u)\bigr)=2(-1)^n\Delta_0(w).
\]
Hence \(\mu(v)-\mu(u)=2\Delta_0(w)\). Setting \(r=0\) gives \(B_0(\theta)=J_0\), so the dihedral point becomes the mixed normalizer slice \((A(\theta),J_0)\).
\end{proof}

\begin{remark}
Combinatorially, \(\mu(x)\) is a signed \(a\)-count: each occurrence of the letter \(a\) contributes \(+1\) or \(-1\) according to whether the number of preceding \(b\)'s is even or odd.
\end{remark}

The dihedral family therefore has two equivalent readings: one through the alternating row-prefix invariant \(\Delta_0(w)\), one through the signed \(a\)-count \(\mu(v)-\mu(u)\). The next proposition records a different algebraic specialization of the same principal family.

\begin{proposition}\label{prop:q-i-collapse}
Fix \(r\in\Z\) and specialize the family to \(T=-1\), equivalently \(q=i\). Then
\begin{equation}\label{eq:trace-q-i-collapse}
\tr\bigl(U_{w,r}(\theta,t)\bigr)\big|_{T=-1}=2\cos\bigl(t\,P_r(-1)\bigr),
\end{equation}
where \(P_r(T)=(T-1)M_w(T,T^r)\). In particular, if \(P_r(-1)\neq 0\), then the choice
\[
t=\frac{\pi}{2|P_r(-1)|}
\]
produces a trace-zero witness on this slice.
\end{proposition}

\begin{proof}
Set
\[
K:=A(\theta)\big|_{q=i}=\mathrm{diag}(i,-i),
\qquad
D_r(\theta)\big|_{q=i}=K^r,
\qquad
B_r(\theta,t)\big|_{q=i}=K^rR(t).
\]
Let
\[
\mathcal R:=\{R(s):s\in\R\}\le SU(2).
\]
This is an abelian subgroup, and a direct calculation gives
\[
K R(s)K^{-1}=R(-s).
\]
Hence every element of \(\mathcal G:=\langle K,\mathcal R\rangle\) can be written as \(K^mR(s)\), with the rotation parameter changing sign whenever a factor is moved across an odd power of \(K\).

Evaluate \(w\) at \((K,K^rR(t))\) and write
\[
U(t):=U_{w,r}(\theta,t)\big|_{T=-1}\in\mathcal G.
\]
Each occurrence of \(b\) or \(b^{-1}\) contributes one factor \(R(\pm t)\), so after all such factors are moved to the right and combined inside the abelian group \(\mathcal R\), one gets
\begin{equation}\label{eq:U-is-rotation-mt}
U(t)=K^{m_0}R(mt)
\end{equation}
for some integers \(m_0,m\).

Now set \(t=0\). Then \(B_r(\theta,0)\big|_{q=i}=K^r\), so both substituted generators are powers of \(K\) and therefore commute. Since \(w\in F_2'\), it follows that
\[
U(0)=I.
\]
Because \(R(0)=I\), equation~\eqref{eq:U-is-rotation-mt} yields \(K^{m_0}=I\). Thus \(m_0\equiv 0 \pmod 4\), and therefore
\[
U(t)=R(mt).
\]
Taking traces gives
\[
\tr(U(t))=\tr(R(mt))=2\cos(mt).
\]

To identify the integer \(m\), note that \(R(mt)=I+mtJ_0+O(t^2)\), so its upper-right entry has derivative \(m\) at \(t=0\). In the notation of Proposition~\ref{prop:trace-quadratic-fox}, this means
\[
\partial_t\beta_{w,r}(T,0)\big|_{T=-1}=m.
\]
By \eqref{eq:beta-prime-fox},
\[
\partial_t\beta_{w,r}(T,0)=-T^r(T-1)M_w(T,T^r)=-T^rP_r(T).
\]
At \(T=-1\), one has \(T^r=(-1)^r=\pm 1\), hence \(m=\pm P_r(-1)\). Since \(\cos\) is even, this gives
\[
\tr(U(t))=2\cos(mt)=2\cos\bigl(tP_r(-1)\bigr),
\]
which is exactly \eqref{eq:trace-q-i-collapse}.
\end{proof}

The term \emph{quaternionic point} refers to the specialization \(T=-1\), equivalently \(q=i\). At that value one has \(K^2=-I\), and together with \(J_0=R(\pi/2)\) the relations \(J_0^2=K^2=-I\) and \(KJ_0=-J_0K\) generate the quaternion group \(Q_8\)~\cite{ConwaySmith03}. Thus the distinguished algebraic point on this slice lies naturally in a quaternionic configuration.

The formula at \(T=-1\) depends only on the specialized value \(P_r(-1)\). Since exact separation is unchanged under automorphisms of \(F_2\), it is enough to obtain a nonzero value at \(T=-1\) for some automorphic representative. This gives the following reformulation of Proposition~\ref{prop:q-i-collapse}.

\begin{proposition}\label{prop:aut-quaternionic}
Let \(w\in F_2'\) and let \(\phi\in\mathrm{Aut}(F_2)\). Fix \(r\in\Z\) and write
\[
P_r^{\phi}(T):=(T-1)M_{\phi(w)}(T,T^r).
\]
If \(P_r^{\phi}(-1)\neq 0\), then there exist \(A,B\in SU(2)\) such that \(\tr(w(A,B))=0\).
\end{proposition}

\begin{proof}
By Proposition~\ref{prop:q-i-collapse}, the condition \(P_r^{\phi}(-1)\neq 0\) gives matrices \(A_0,B_0\in SU(2)\) such that
\[
\tr\bigl(\phi(w)(A_0,B_0)\bigr)=0.
\]
Set
\[
A:=\phi(a)(A_0,B_0),
\qquad
B:=\phi(b)(A_0,B_0).
\]
Then \(A,B\in SU(2)\), and functoriality of word evaluation gives
\[
w(A,B)
=w\bigl(\phi(a)(A_0,B_0),\phi(b)(A_0,B_0)\bigr)
=\phi(w)(A_0,B_0).
\]
Taking traces yields \(\tr(w(A,B))=0\).
\end{proof}

\subsection{Local and supplementary witnesses}

While the specialized algebraic points developed in Subsection~\ref{subsec:algebraic-points} provide strong global criteria, they do not explicitly describe the local combinatorial shape of a hard positive-word pair. The present subsection develops more targeted forcing mechanisms to address that gap, beginning with a complete analysis of the first nontrivial local regime.

\subsubsection{The case \texorpdfstring{\(|\Jset|=1\)}{|\Jset|=1}: classification and separation}\label{subsec:J1}

When \(|\Jset|=1\), the two positive words differ in exactly one local move: a block \(a^d\) crosses a single adjacent \(b\). The next proposition combines the resulting classification with the corresponding exact separation statement.

\begin{proposition}\label{prop:J1-separated}
Let \(u\neq v\in\{a,b\}^{\ast}\) be positive words with \(\ab(u)=\ab(v)\), set \(w=u^{-1}v\), and assume \(\Jset=\{j_0\}\). Then \(w\) is conjugate in \(F_2\) to \([a^d,b]^{\varepsilon}\) for some integer \(d\ge 1\) and some sign \(\varepsilon\in\{\pm1\}\). In particular, there exist \(A,B\in SU(2)\) such that \(\tr(w(A,B))=0\).
\end{proposition}

\begin{proof}
Write
\[
u=a^{x_0}ba^{x_1}b\cdots ba^{x_n},
\qquad
v=a^{y_0}ba^{y_1}b\cdots ba^{y_n},
\]
where \(n=\#_b(u)=\#_b(v)\). For \(x\in\{u,v\}\), set \(A_x(0):=0\) and \(A_x(n+1):=\#_a(x)\). Then
\[
A_x(t+1)-A_x(t)=x_t\qquad (0\le t\le n).
\]
Since \(\Jset=\{j_0\}\), one has \(A_u(j)=A_v(j)\) for every \(j\in\{1,\dots,n\}\setminus\{j_0\}\). The endpoint conventions give \(A_u(0)=A_v(0)=0\), and equality of the total \(a\)-counts gives \(A_u(n+1)=A_v(n+1)\). Thus
\[
\delta:=A_v(j_0)-A_u(j_0)\neq 0
\]
is the unique discrepancy among the extended prefix data. Using \(A_x(t+1)-A_x(t)=x_t\), one immediately obtains
\[
y_t=x_t \quad (t\notin\{j_0-1,j_0\}),
\qquad
y_{j_0-1}-x_{j_0-1}=\delta,
\qquad
y_{j_0}-x_{j_0}=-\delta.
\]
Let \(d:=|\delta|\). If \(\delta>0\), then a block \(a^d\) has moved from the \(j_0\)-th \(a\)-run of \(u\) to the preceding one, so for suitable words \(P,Q\),
\[
u=P\,b\,a^d\,Q,
\qquad
v=P\,a^d\,b\,Q.
\]
A direct multiplication gives
\[
w=(Pba^dQ)^{-1}(Pa^dbQ)=Q^{-1}[a^d,b]Q.
\]
If \(\delta<0\), then the same argument with \(u\) and \(v\) interchanged yields
\[
u=P\,a^d\,b\,Q,
\qquad
v=P\,b\,a^d\,Q,
\]
and therefore
\[
w=(Pa^dbQ)^{-1}(Pba^dQ)=Q^{-1}[a^d,b]^{-1}Q.
\]
Thus in all cases \(w=Q^{-1}[a^d,b]^{\varepsilon}Q\) for some \(\varepsilon\in\{\pm1\}\).

Choose
\[
\theta:=\frac{\pi}{4d},
\qquad
A:=A(\theta)=\begin{pmatrix}e^{i\theta}&0\\0&e^{-i\theta}\end{pmatrix},
\qquad
B:=J_0=\begin{pmatrix}0&1\\-1&0\end{pmatrix}.
\]
Then \(A^d=\mathrm{diag}(e^{i\pi/4},e^{-i\pi/4})\) and \(J_0 A^d J_0^{-1}=A^{-d}\), so
\[
[a^d,b](A,B)=A^{-d}B^{-1}A^dB=A^{-2d}=\mathrm{diag}(e^{-i\pi/2},e^{i\pi/2}).
\]
Therefore \(\tr([a^d,b](A,B))=0\). Since trace is invariant under conjugation and satisfies \(\tr(M^{-1})=\tr(M)\) for \(M\in SU(2)\), the same pair \((A,B)\) also gives \(\tr(w(A,B))=0\).
\end{proof}

This classification completely resolves the first nontrivial local discrepancy. For hard pairs beyond this direct reduction, we conclude the certified menu with two supplementary mechanisms, one exploiting binary-dihedral symmetry and the other working in Fricke--Vogt trace coordinates.

\subsubsection{Character coordinates and further explicit tests}\label{subsec:char-coords}

These final components enlarge the certified witness library in two complementary directions. We begin with an elementary matrix slice that exploits binary-dihedral symmetry and then turn to a criterion formulated in Fricke--Vogt trace coordinates.

\begin{proposition}\label{prop:dihedral-sieve}
Let
\[
D_\infty=\langle s,t\mid s^2=t^2=1\rangle,
\qquad r:=st,
\]
and let \(\rho:F_2\to D_\infty\) be the homomorphism defined by \(\rho(a)=s\), \(\rho(b)=t\). For each \(w\in F_2'\), there exists a unique integer \(\kappa(w)\in\Z\) such that
\[
\rho(w)=r^{2\kappa(w)}.
\]
Consequently, for every pair \(A,B\in SU(2)\) satisfying \(A^2=B^2=-I\), one has
\[
w(A,B)=\pm(AB)^{2\kappa(w)}.
\]
In particular, if \(\kappa(w)\neq 0\), then \(w\) admits a trace-zero witness in \(SU(2)\).
\end{proposition}

\begin{proof}
The commutator subgroup of the infinite dihedral group \(D_\infty\) is \(\langle r^2\rangle\), an infinite cyclic group. Since \(w\in F_2'\), its image under \(\rho\) lies in \([D_\infty,D_\infty]\), so there exists a unique integer \(\kappa(w)\) such that
\[
\rho(w)=r^{2\kappa(w)}.
\]
This integer depends only on \(w\), not on any chosen matrix pair.

Now let \(A,B\in SU(2)\) satisfy \(A^2=B^2=-I\), and let \(\bar A,\bar B\in SO(3)\) be their images under the double cover \(SU(2)\to SO(3)\). Then \(\bar A^2=\bar B^2=1\), so the assignment \(s\mapsto \bar A\), \(t\mapsto \bar B\) defines a homomorphism
\[
\psi_{A,B}:D_\infty\longrightarrow \langle \bar A,\bar B\rangle\le SO(3).
\]
Applying \(\psi_{A,B}\) to the identity \(\rho(w)=r^{2\kappa(w)}\) gives
\[
w(\bar A,\bar B)=\psi_{A,B}(\rho(w))=(\bar A\bar B)^{2\kappa(w)}.
\]
Lifting from \(SO(3)\) back to \(SU(2)\) yields
\[
w(A,B)=\pm(AB)^{2\kappa(w)}.
\]

For the final statement, consider the family
\[
J_0=\begin{pmatrix}0&1\\-1&0\end{pmatrix},
\qquad
B(\theta)=\begin{pmatrix}0&e^{i\theta}\\-e^{-i\theta}&0\end{pmatrix}
\qquad (0\le \theta\le \pi).
\]
Then \(J_0^2=B(\theta)^2=-I\) and
\[
J_0B(\theta)=-\begin{pmatrix}e^{-i\theta}&0\\0&e^{i\theta}\end{pmatrix}.
\]
Therefore
\[
(J_0B(\theta))^{2\kappa(w)}=\begin{pmatrix}e^{-2i\kappa(w)\theta}&0\\0&e^{2i\kappa(w)\theta}\end{pmatrix},
\]
and so
\[
\tr\bigl(w(J_0,B(\theta))\bigr)=\pm\,\tr\bigl((J_0B(\theta))^{2\kappa(w)}\bigr)=\pm 2\cos\bigl(2\kappa(w)\theta\bigr).
\]
If \(\kappa(w)\neq 0\), choosing \(\theta=\pi/(4|\kappa(w)|)\) gives \(\tr(w(J_0,B(\theta)))=0\). Hence \(w\) has a trace-zero witness.
\end{proof}

A complementary viewpoint is to work instead with trace coordinates. The classical description of the rank-two \(SU(2)\) character variety turns the search for a trace-zero witness into a sign problem for a single polynomial on a compact region. This viewpoint yields additional direct criteria that sit naturally beside the slice-based tests.

For \(A,B\in SU(2)\), set
\[
x:=\tr(A),\qquad y:=\tr(B),\qquad z:=\tr(AB).
\]
The classical Fricke--Vogt trace identities imply that for every word \(w\in F_2\) there exists a unique polynomial
\[
P_w(x,y,z)\in \Z[x,y,z]
\]
such that
\[
P_w\bigl(\tr(A),\tr(B),\tr(AB)\bigr)=\tr\bigl(w(A,B)\bigr)
\]
for all \(A,B\in SU(2)\); see Goldman~\cite{Goldman09}.

By the classical Fricke--Vogt trace theorem for two-generator \(SL_2\)-representations, the map
\[
SU(2)^2\longrightarrow \R^3,
\qquad
(A,B)\longmapsto \bigl(\tr(A),\tr(B),\tr(AB)\bigr)
\]
has image
\[
D:=\Bigl\{(x,y,z)\in[-2,2]^3:\ 0\le x^2+y^2+z^2-xyz\le 4\Bigr\},
\]
and this region is compact and connected; see Goldman~\cite[\S2]{Goldman09}. Consequently, if \(w\in F_2'\) and its trace polynomial \(P_w\) takes a nonpositive value at some point of \(D\), then connectedness and the intermediate value theorem yield a point of \(D\) where \(P_w=0\), hence a trace-zero witness in \(SU(2)^2\).

Evaluating the trace polynomial at a few explicit interior points yields an additional immediate criterion. This criterion complements the slice-based constructions and will also enter the finite-menu summary in Section~\ref{sec:obstructions}.

\begin{proposition}\label{prop:interior-point-criteria}
Let \(w\in F_2'\), and let \(P_w(x,y,z)\) be its trace polynomial. If at least one of the two interior evaluations
\[
P_w(0,1,-1),\qquad P_w(0,1,0)
\]
is nonpositive, then \(w\) admits a trace-zero witness.
\end{proposition}

\begin{proof}
A direct check shows that both points lie in the interior of \(D\):
\[
0^2+1^2+(-1)^2-0=2,\qquad 0^2+1^2+0^2-0=1,
\]
and both values lie strictly between \(0\) and \(4\). The claim therefore follows from the Fricke--Vogt continuity principle stated at the beginning of this subsection.
\end{proof}

Having established the individual witness constructions, Section~\ref{sec:obstructions} synthesizes them into a certified menu, isolates the residual class left untreated by it, and explains why any universal completion cannot arise from a fixed finite-image method.

\section{Synthesis: the residual class and the limits of finite-image methods}\label{sec:obstructions}

The explicit constructions in Section~\ref{sec:witnesses} successfully separate a broad range of hard word pairs. Yet a universal answer requires understanding the boundaries of these mechanisms. This section therefore shifts the perspective from geometric construction to structural obstruction. We prove that no fixed finite-image strategy can be universal, consolidate the results proved earlier into a certified menu, and isolate the super-degenerate residual class that sharply defines the present frontier of the separation conjecture.

\subsection{Algorithmic limits of fixed finite-image tests}\label{subsec:finite-menu-slices}

We begin with the negative result that sets the scope of the constructive framework developed earlier. The point is not that explicit geometric families are ineffective, but that no \emph{fixed} finite-image test can be universal. More precisely, the theorem below shows that no predetermined finite collection of finite-image tests can distinguish all hard positive pairs. This explains why the slice-based viewpoint remains essential: the families may be explicit and finite in number, but the successful specialization is allowed to depend on the word.

\begin{theorem}\label{thm:pigeonhole-finite-criterion}
Let $\rho_i:F_2\to G_i$ ($i=1,\dots,m$) be group homomorphisms such that each image $\mathrm{Im}(\rho_i)$ is finite.
Let $\Phi:\{a,b\}^\ast\to \prod_{i=1}^m \mathrm{Im}(\rho_i)$ be the map obtained by restricting each $\rho_i$ to the positive monoid
$\{a,b\}^\ast\subset F_2$ and taking the product,
\[
\Phi(x)=(\rho_1(x),\dots,\rho_m(x)).
\]
Then there exist distinct positive words $u\neq v$ with $\mathrm{ab}(u)=\mathrm{ab}(v)$ and $\Phi(u)=\Phi(v)$.
Equivalently, $w:=u^{-1}v\in F_2'\setminus\{1\}$ lies in $\bigcap_{i=1}^m \ker(\rho_i)$.
\end{theorem}

In particular, no fixed finite family of finite-image tests can distinguish all positive hard pairs.
\begin{proof}
Let $M:=\prod_{i=1}^m|\mathrm{Im}(\rho_i)|<\infty$.
For fixed integers $n\ge 1$ and $0\le k\le n$, the number of positive words of length $n$ with exactly $k$ occurrences of $a$ is $\binom{n}{k}$.
Choose \(n\) large enough that \(\binom{n}{\lfloor n/2\rfloor}>M\), and set \(k=\lfloor n/2\rfloor\).
Among those words, the map \(\Phi\) takes at most \(M\) values, so two distinct words \(u\neq v\) must satisfy \(\Phi(u)=\Phi(v)\).

These two words have the same number of \(a\)'s and the same number of \(b\)'s, hence \(\mathrm{ab}(u)=\mathrm{ab}(v)\) and therefore \(w:=u^{-1}v\in F_2'\).
Also, for every \(i\),
\[
\rho_i(w)=\rho_i(u)^{-1}\rho_i(v)=1,
\]
so \(w\in\bigcap_{i=1}^m\ker(\rho_i)\).
Finally, distinct positive words are distinct reduced words in the free group, so \(w\neq 1\).
\end{proof}

Theorem~\ref{thm:pigeonhole-finite-criterion} identifies a basic structural obstruction: no fixed catalog of finite-image tests can yield a universal separation criterion. This limitation is exactly what makes the slice-driven approach necessary. Rather than searching through predetermined rigid quotients, we work with a finite menu of explicit continuous families, with the successful geometric specialization allowed to depend on the word itself. In this way, the combinatorial invariants extracted earlier guide the choice of slice, and the constructive results of the previous sections bypass the intrinsic limitation of fixed finite-subgroup methods.

\subsection{The current certified slice menu and its range}\label{subsec:current-menu}

We now record the present witness library in its formal synthesized form. The constructions proved earlier produce a small collection of explicit slice families and algebraic points on those families. Each item below is a constructive criterion: it extracts a first-order invariant from the positive-word-difference presentation and uses that invariant to pick a slice, or a distinguished point on a slice, where trace zero is forced.

\begin{theorem}\label{thm:certified-menu-items}
Let \(u\neq v\in\{a,b\}^{\ast}\) satisfy \(\ab(u)=\ab(v)\), and set \(w:=u^{-1}v\in F_2'\). Each of the following conditions implies that \(u\) and \(v\) are exactly separable in the standard two-state measure-once quantum finite automaton model, equivalently that there exist \(A,B\in SU(2)\) with \(\tr(w(A,B))=0\):
\begin{enumerate}[label=\emph{(\alph*)}, leftmargin=2.2em]
\item \textbf{Dihedral / signed \(a\)-count criterion.} Equivalently, the alternating row-prefix invariant \(\Delta_0(w)\) is nonzero, or the signed \(a\)-count satisfies \(\mu(u)\neq\mu(v)\).

\item \textbf{Quaternionic point on the slope-visible slice.} For some slope \(r\in\Z\), the specialized value \(P_r(-1)=(T-1)M_w(T,T^r)|_{T=-1}\) is nonzero.

\item \textbf{Local one-row regime.} The active row set satisfies \(|\Jset|=1\).
\end{enumerate}
\end{theorem}

The synthesis in Theorem~\ref{thm:certified-menu-items} highlights that exact separation is driven by forcing mechanisms rather than disjoint word categories. While the three criteria are extracted from the ordered positive-word-difference presentation \(w=u^{-1}v\), they collectively represent a unified slice-based strategy. Cases~\emph{(a)} and~\emph{(b)} use global symmetries at distinguished algebraic points of the slope-visible family, while case~\emph{(c)} exploits the local commutator structure of the one-row regime. Their proofs, derived from Propositions~\ref{prop:dihedral-point-trace}, \ref{prop:q-i-collapse}, and~\ref{prop:J1-separated}, demonstrate that trace-zero witnesses can be systematically recovered from the combinatorial invariants \(\Delta_0(w)\), \(P_r(-1)\), and \(\Jset\).

Crucially, the effectiveness of this menu does not contradict the algorithmic limits established in Theorem~\ref{thm:pigeonhole-finite-criterion}. Although a successful specialization often lands in a finite dihedral or quaternionic subgroup, the order of the relevant group is not fixed; instead, it is dynamically dictated by the word's own invariants. By allowing the geometric parameters to vary with the word, this slice-driven framework bypasses the structural obstruction that restricts fixed finite-image methods.

Alongside Theorem~\ref{thm:certified-menu-items}, Proposition~\ref{prop:interior-point-criteria} provides a separate interior-point criterion on the trace region \(D\), and Proposition~\ref{prop:dihedral-sieve} provides an additional certified binary-dihedral slice in which both generators satisfy \(A^2=B^2=-I\). A companion numerical experiment applies a fixed empirical first-hit menu to \(50{,}000\) random hard pairs. The corresponding first-hit counts for the mixed branch, the dihedral branch, the \(T=-1\) filter, and the four-point torsion search are \(45{,}486\), \(2{,}965\), \(1{,}233\), and \(316\), with no misses. This supports the practical reach of the slice-based search strategy on the sampled random model, while the certified mathematical coverage is given by the theorems and propositions above. The reproduction materials are publicly available (see the Data and Code Availability statement).

\subsection{The residual class, the stronger finite-subgroup obstruction, and the completion target}\label{subsec:residual}

The previous subsection records what the present slice menu can already certify. The complementary task is to identify the residual class that escapes these tests and to clarify which further strategies are excluded on structural grounds. The next theorem first strengthens the negative side by showing that even a much broader finite-subgroup search cannot be universal.

\begin{theorem}\label{thm:finite-subgroup-law}
Let
\[
w_\star\ :=\ [a^2,b^2]^{120}\in F_2'.
\]
Then:
\begin{enumerate}
\item $w_\star\neq 1$ in the free group $F_2$.
\item For every finite subgroup $H\le SU(2)$ and every $A,B\in H$, one has $w_\star(A,B)=I$ and hence $\operatorname{tr}(w_\star(A,B))=2$.
\end{enumerate}
Consequently, no method based only on finitely many evaluations inside finite subgroups of \(SU(2)\) can be universal.
\end{theorem}

\begin{proof}
We first show that \(w_\star\) is trivial on every finite subgroup of \(SU(2)\). Finite subgroups of \(SU(2)\) are cyclic groups, binary dihedral groups \(Q_{4n}\), and the three exceptional groups \(2A_4\), \(2S_4\), and \(2A_5\), see for example~\cite{Serre77,ConwaySmith03}.

\smallskip
\noindent\emph{Cyclic case.}
If \(H\) is cyclic, then \(H\) is abelian, so \([A^2,B^2]=I\) for all \(A,B\in H\).
Hence \(w_\star(A,B)=I\).

\smallskip
\noindent\emph{Binary dihedral case.}
Write
\[
Q_{4n}=\langle x,y\mid x^{2n}=1,\ y^2=x^n,\ y^{-1}xy=x^{-1}\rangle.
\]
Every element of \(Q_{4n}\) has the form \(x^k\) or \(yx^k\).
Its square is always in the cyclic subgroup \(\langle x\rangle\). Therefore \(A^2,B^2\in\langle x\rangle\), so they commute and \([A^2,B^2]=I\).
Again \(w_\star(A,B)=I\).

\smallskip
\noindent\emph{Exceptional cases.}
Let \(H\in\{2A_4,2S_4,2A_5\}\), and set \(C=[A^2,B^2]\in H\).
The images of these groups in \(SO(3)\) are \(A_4\), \(S_4\), and \(A_5\).
Their exponents are
\[
\exp(A_4)=6,\qquad \exp(S_4)=12,\qquad \exp(A_5)=30.
\]
If \(x\in H\) projects to an element of order \(t\) in the corresponding rotation group, then \(x^t\in\{\pm I\}\), since the kernel of \(H\to SO(3)\) is \(\{\pm I\}\).
Hence \(x^{2t}=I\).
It follows that every element of \(2A_4\) has order dividing \(12\), every element of \(2S_4\) has order dividing \(24\), and every element of \(2A_5\) has order dividing \(60\).
In particular, every element of \(H\) has order dividing
\[
\mathrm{lcm}(12,24,60)=120.
\]
Therefore \(C^{120}=I\), so \(w_\star(A,B)=C^{120}=I\).

This proves part~(2).

\smallskip
To prove part~(1), it is enough to exhibit one evaluation for which \(w_\star\) is not the identity.
Take
\[
A=\begin{pmatrix}1&1\\0&1\end{pmatrix},\qquad
B=\begin{pmatrix}1&0\\1&1\end{pmatrix}\in SL(2,\mathbb Z).
\]
Then
\[
A^2=\begin{pmatrix}1&2\\0&1\end{pmatrix},\qquad
B^2=\begin{pmatrix}1&0\\2&1\end{pmatrix},
\]
and a direct calculation gives
\[
[A^2,B^2]=A^{-2}B^{-2}A^2B^2=
\begin{pmatrix}21&8\\-8&-3\end{pmatrix}.
\]
This matrix is not the identity, and its trace is \(18\), so it has infinite order in \(SL(2,\mathbb Z)\).
Hence \([A^2,B^2]^{120}\neq I\).
Therefore \(w_\star\neq 1\) in the free group.
\end{proof}

The theorem above shows that finite subgroups are not enough for a universal criterion, but it does not imply that \(w_\star\) or similar words fail to admit trace-zero witnesses in \(SU(2)\). In fact the opposite is true: commutator powers admit such witnesses uniformly. The argument is short in \(SU(2)\), so it is convenient to record it directly.

\begin{proposition}\label{prop:trace-zero-commutator-powers}
Fix integers \(m,n,k\ge 1\) and consider the word \(w=[a^m,b^n]^k\in F_2'\).
Then there exist \(A,B\in SU(2)\) such that \(\operatorname{tr}(w(A,B))=0\).
In particular, \(w_\star=[a^2,b^2]^{120}\) admits trace-zero witnesses in \(SU(2)\).
\end{proposition}

\begin{proof}
Let
\[
D(\phi):=\begin{pmatrix}e^{i\phi}&0\\0&e^{-i\phi}\end{pmatrix},
\qquad
J_0:=\begin{pmatrix}0&1\\-1&0\end{pmatrix},
\qquad
R(t):=\begin{pmatrix}\cos t&\sin t\\-\sin t&\cos t\end{pmatrix}=\exp(tJ_0).
\]
Set \(\theta=\pi/(2k)\) and \(g:=D(\theta)\). Then
\[
\operatorname{tr}(g^k)=2\cos(k\theta)=2\cos(\pi/2)=0.
\]
A direct computation gives
\[
[J_0,D(\theta/2)]=D(\theta)=g.
\]
Now set
\[
A:=R\!\left(\frac{\pi}{2m}\right),
\qquad
B:=D\!\left(\frac{\theta}{2n}\right).
\]
Since \(R(\pi/2)=J_0\), it follows that \(A^m=J_0\). Clearly \(B^n=D(\theta/2)\). Therefore
\[
[A^m,B^n]=[J_0,D(\theta/2)]=g,
\]
and hence
\[
w(A,B)=[A^m,B^n]^k=g^k,
\]
which has trace zero.
\end{proof}

It is also natural to isolate the hard pairs that remain invisible to the present certified menu.

\begin{definition}\label{def:super-degenerate}
A hard positive-word difference \(w=u^{-1}v\in F_2'\) is called \emph{super-degenerate} if it satisfies
\[
\Delta_0(w)=0,
\qquad
P_r(-1)=0\ \text{for both parities of }r,
\]
\[
|\Jset|\neq 1,
\qquad
\kappa(w)=0,
\qquad
P_w(0,1,-1)>0,
\qquad
P_w(0,1,0)>0.
\]
\end{definition}

If \(w\) is super-degenerate, then Proposition~\ref{prop:dihedral-point-trace} gives \(\tr(w)=2\) at the dihedral point of the slope-visible family, Proposition~\ref{prop:q-i-collapse} gives \(\tr(w)=2\) at the quaternionic point for every slope parity, and Corollary~\ref{cor:mu-delta0-equivalence} shows that the mixed normalizer exponent also vanishes. Proposition~\ref{prop:J1-separated} does not apply because \(|\Jset|\neq 1\), Proposition~\ref{prop:dihedral-sieve} does not apply because \(\kappa(w)=0\), and Proposition~\ref{prop:interior-point-criteria} does not apply because both listed interior evaluations are positive. The present certified menu therefore leaves this residual class untouched. Any completion must add a genuinely new forcing mechanism.

Taken together, these results rule out any completion based only on a fixed finite list of finite quotients or finite subgroups. At the same time, they show that the explicit slice menu already covers a very large portion of the problem. What remains logically viable is a finite menu of explicit analytic slice families, each equipped with a connected compact parameter space and a distinguished basepoint at which the two generators commute. This is the completion target toward which the earlier constructions point.

\begin{conjecture}\label{conj:finite-menu-cover}
There exist finitely many explicit families \(\mathscr S_1,\dots,\mathscr S_M\) with the following properties.
For each \(i\), the family is given by a continuous map
\[
(A_i,B_i):X_i\longrightarrow SU(2)^2
\]
from a compact connected parameter space \(X_i\), together with a distinguished point \(x_i^0\in X_i\) such that \(A_i(x_i^0)\) and \(B_i(x_i^0)\) commute.
For every hard word \(w=u^{-1}v\in F_2'\) arising from distinct positive words, there exist \(i\) and \(t\in X_i\) such that
\[
\tr\bigl(w(A_i(t),B_i(t))\bigr)\le 0.
\]
\end{conjecture}

Because \(w\in F_2'\), the commuting basepoint condition implies
\[
\tr\bigl(w(A_i(x_i^0),B_i(x_i^0))\bigr)=2.
\]
Hence continuity on the connected space \(X_i\) shows that any parameter value with nonpositive trace forces a trace-zero witness somewhere on the same family.

The menu items already proved in this paper show that such a program is not empty formalism. The slope-visible family, its algebraic endpoints, the mixed normalizer slice, the binary-dihedral slice from Proposition~\ref{prop:dihedral-sieve}, and the trace-coordinate tests all provide concrete components of the desired library. The remaining task is to construct an additional analytic family that detects the residual class from Definition~\ref{def:super-degenerate}.

\section{Conclusion and outlook}\label{sec:conclusion}

The paper develops the exact-separation problem through a local slice philosophy. Starting from the row-wise metabelian structure of a positive-word difference, it extracts slope specializations that remain nontrivial and places them on explicit low-dimensional slice families in \(SU(2)^2\). The main payoff is that the trace-zero problem becomes visible on low-dimensional slices through directly computable invariants: the alternating row-prefix quantity \(\Delta_0(w)\), equivalently the signed \(a\)-count on the \(r=0\) specialization, the quaternionic specialization \(P_r(-1)\), the one-row regime \(|\Jset|=1\), the binary-dihedral exponent \(\kappa(w)\), and selected interior evaluations of the trace polynomial. This yields a small certified menu that already covers a wide range of hard positive-word differences while keeping the geometry explicit.

Just as importantly, the paper clarifies the remaining obstruction. Fixed finite-image methods are ruled out in principle, yet the unresolved class is sharply delimited by the super-degenerate conditions. The problem therefore condenses to a concrete completion target: construct additional explicit analytic slice families whose trace geometry detects that residual class. In this sense, the present work reduces the two-state exact-separation conjecture from a diffuse search over \(SU(2)^2\) to a focused geometric extension problem centered on the residual class and Conjecture~\ref{conj:finite-menu-cover}.

\nonumsection{Declarations}

\noindent\textbf{Data and Code Availability}\par
The study uses no external empirical dataset. The numerical experiment reported in the paper is based on \(50{,}000\) randomly generated hard positive-word pairs and records the corresponding first-hit counts for the fixed empirical witness menu. The source code, generation protocol, and summary output needed to reproduce these computations are publicly available in the \href{https://github.com/chen8965/exact-separation-of-words-via-trace-geometry-experiment}{project repository}.

\bibliographystyle{unsrt}
\bibliography{ref}

\end{document}